\documentclass[11pt]{article}
\usepackage{latexsym}
\usepackage{amssymb, amsmath, xspace, lscape,  latexsym}
\usepackage{amsthm}
\usepackage{graphicx}
\usepackage{tikz}

\newcommand{\bea}{\begin{eqnarray}}
\newcommand{\eea}{\end{eqnarray}}
\def\beq#1#2\eeq{
        \begin{equation}
        \label{#1}
            #2
        \end{equation}}

\newcommand{\al}{\alpha}
\newcommand{\la}{\lambda}

\newcommand{\C}{\mathbb C}

\newcommand{\rw}{\rightarrow}

\def\btheor#1\etheor{
        \begin{theor}
            #1
        \end{theor}
    }

    \def\bsled#1\esled{
        \begin{sled}
            #1
        \end{sled}   }

\newcommand{\bq}{\begin{equation}}
\newcommand{\eq}{\end{equation}}
\newcommand{\bbq}{\begin{equation*}}
\newcommand{\eeq}{\end{equation*}}

\newtheorem{rem}{Remark}

\newtheorem{theorem}{Theorem}

\newtheorem{lemma}{Lemma}

\newtheorem{cor}{Corollary}

\def\hm#1{#1\nobreak\discretionary{}{\hbox{\m@th$#1$}}{}}
\def\mi#1{\discretionary{\hbox{\m@th$#1$}}{\hbox{\m@th$#1$}}{}}

\vspace{4ex}
\begin{document}
\title{\bf Hankel determinants for a perturbed Laguerre weight and Painlev\'e \uppercase\expandafter{\romannumeral5} equation}
\author{Min Chen$^{1}$, Yang Chen$^{2}$,  En-Gui Fan$^{1}$\thanks{Corresponding author: faneg@fudan.edu.cn}\\
{\normalsize $^{1}$School of Mathematical Science, Fudan University, Shanghai 200433, P.R. China}\\
{\normalsize $^{2}$Faculty of Science and Technology, Department of Mathematics, University of Macau}
       }
\date{\today}
\maketitle
\begin{abstract}
In this paper, we study Hankel determinants generated from the following perturbed Laguerre weight function,
\bbq
w(x,t)=(x+t)^{\la}x^{\al}e^{-x},\;\; t>0,\; \al>0,\; \al+\la+1>0,\;\; x\in [0,\infty).
\eeq
Under the double scaling scheme, the random matrix size $n\rw \infty,$ $t\rw 0,$ such that $s=4nt$ is finite, we give the uniform asymptotic
approximations of Hankel determinants in terms of a solution of a third-order nonlinear differential equation, which is equivalent to a particular
Painlev\'e \uppercase\expandafter{\romannumeral5} equation (${\rm P_{\uppercase\expandafter{\romannumeral5}}}$). In fact, this ${\rm P_{\uppercase\expandafter{\romannumeral5}}}$ equation is equivalent to the general
Painlev\'e \uppercase\expandafter{\romannumeral3} equation. The asymptotic approximations of the leading coefficients and the recurrence coefficients for the corresponding orthogonal polynomials
also involve the ${\rm P_{\uppercase\expandafter{\romannumeral5}}}$ equation. The asymptotic   analysis is based on our earlier results by using the Deift-Zhou nonlinear steepest descent method.

\end{abstract}
\noindent

\newpage

\setcounter{equation}{0}
\section{Introduction and statement of results}

We consider Hankel determinants are defined by
\bq\label{B02}
D_{n}(t,\al,\la)=\det\left(\int_{\mathbb{R}_{+}}x^{j+k}w(x,t)dx\right)_{j,k=0}^{n-1},\;\;\; \mathbb{R}_{+}=(0,\infty),
\eq
where $w(x,t)$ is the perturbed Laguerre weight as follows
\bq\label{B01}
w(x,t)=(x+t)^{\la}x^{\al}e^{-x},\;\; t>0,\; \al>0,\; \al+\la+1>0,\;\; x\in [0,\infty).
\eq
\par
It is well-known that the above Hankel determinants can also be expressed by
\bq\label{B03}
D_{n}(t,\al,\la)=\prod_{j=0}^{n-1}h_{j}(t)=\frac{1}{n!}Z_{n}(t),
\eq
where $Z_{n}(t)$ are known as the partition function and its multiple integral representation is given by
\bq\label{B04}
Z_{n}(t)=\int_{\mathbb{R}_{+}^{n}}\prod_{1\leq{j<k}\leq{n}}\left(x_{j}-x_{k}\right)^{2}\prod_{\ell=1}^{n}w(x_{\ell},t)dx_{\ell},
\eq
and $\{h_{j}(t): j=0, 1, \ldots, n-1\}$ are the squares of the $L^{2}$ norm of the monic polynomials $\pi_{n}(x)$ orthogonal with respect to the weight function $w(x, t)$ in (\ref{B01}), namely,
\bq\label{B05}
\int_{\mathbb{R}_{+}}\pi_{j}(x)\pi_{k}(x)w(x,t)dx=h_{j}(t)\delta_{jk}.
\eq
Moreover, the above monic orthogonal polynomials satisfy the following three-term recurrence relation,
\bq\label{B11}
x\pi_{n}(x)=\pi_{n+1}(x)+\al_{n}(t)\pi_{n}(x)+\beta_{n}(t)\pi_{n-1}(x).
\eq
It is obviously that the recurrence coefficients $\al_{n}(t)$ and $\beta_{n}(t)$ depend on $t.$
\par
The analysis of Hankel determinant of a random matrix model is an important object in random matrix theory. One finds that it has many important applications in other fields, for instance, a singularly perturbed Hankel determinant was considered by Osipov and Kanzieper \cite{OVAKE2007} from the point of bosonic replica fields theories; it is also considered by Chen and Its \cite{ChenIts2010} and is motivated in part by a finite temperature integrable quantum field theory; the work of Texier and Majumdar \cite{TM2013} described that Hankel determinant may be interpreted as the generating function for the distribution function of the Wigner delay time in chaotic cavities.
\par
There are many methods to study Hankel determinant. Deift, Its and Krasovsky applied the nonlinear steepest descent method for RH problem to study the asymptotics of $n$ dimensional Toeplitz determinants with Fisher-Hartwig singularities and it also provides a powerful way to the analysis of Hankel determinants in \cite{DIK2011}. Adler and Van Moerbeke used a multi-time approach to derive differential equations governing the logarithmic derivatives of Hankel determinants in \cite{AM2001}. For the more information and details of the Deift-Zhou nonlinear steepest descent method, we refer \cite{DKV1999, DX1999}.  The connection between the Hankel determinants isomonodromic $\tau$-function generated by general semi-classical weights and Painlev\'e equations following abstractly by Jimbo, Miwa and Ueno in \cite{JMU1981}, and these are established by Bertola and his collaborators in \cite{BEH2006, B2009}.
\par
We also note that a more general weight function similar with (\ref{B01}), due to there isn't any restriction of the parameter $\la,$ namely, $\la$ is a ``free'' parameter, which appears the characterization of Shannon capacity in the study of the outage and error probability in wireless communication, see \cite{YangMcky2012}. The ``free'' parameter $\la$ ``generates'' the Shannon capacity and $t$ plays the role of the time variable in the ensuring ${\rm P_{\uppercase\expandafter{\romannumeral5}}}$ equation in \cite{YangMcky2012}, and see \cite{BasorChen} for more information. It is worth to point out that the restrictions of $\al$ and $\la$ in the weight function in (\ref{B01}) are necessary conditions to follow the uniform approximation of the corresponding orthogonal polynomials, the more information can be found in \cite{CFC2017}.
\par
It is interesting that the regularly perturbed Laguerre weight in (\ref{B01}) can ``degenerate'' a singularly perturbed Laguerre weight. Heruistically, it maybe seen as follows (s=4nt)
\bbq
(x+t)^{\la}x^{\al}e^{-x}=x^{\al+\la}\left[\left(1+\frac{s}{4nx}\right)^{\frac{4nx}{s}}\right]^{\frac{\la{t}}{x}}e^{-x}\rightarrow x^{\al+\la}e^{-x+\frac{\la{t}}{x}},\;\; n\rightarrow \infty.
\eeq
Xu, Dai and Zhao used the RH approach to study the double scaling limit of the correlation function kernel associated to the singularly perturbed Laguerre weight in \cite{XDZ2014}, and the limiting kernel is a particular ${\rm P_{\uppercase\expandafter{\romannumeral3}}}$ kernel which can also translate to Bessel kernel and Airy kernel in certain conditions, respectively. They also showed a uniform asymptotic approximation of the corresponding Hankel determinants in \cite{XDZ2015}. Mezzadri and Mo \cite{MM2009}, Brightmore, Mezzadri and Mo \cite{BMM2015} considered the partition function generated by a singularly perturbed Gaussion weight, for special case, a particular ${\rm P_{\uppercase\expandafter{\romannumeral3}}}$ equation is also derived therein.
\par
In our pervious paper \cite{CFC2017}, we are interesting in a phase transition phenomenon as $t\rightarrow 0.$ In order to study this phenomenon, we introduced a double scaling scheme, let $n\rw \infty,$ $t\rw 0,$ such that $s=4nt$ is finite, and adopted the Deift-Zhou nonlinear steepest descent method for the corresponding RH problems. We found a new limit of the corresponding eigenvalue correlation function. It is involved a third-order nonlinear differential equation, which is integrable and its Lax pair is also found. It is worth mentioning that this Lax pair do not contain in \cite{FIK2006}, and it also seems as a supplement and perfect the Lax pair first prompted by Kapaev and Hubert \cite{KH1999} because there are two auxiliary functions in their situation, but we just need one auxiliary function and its derivatives to express the other quantities in the Lax pair. Moreover, both of the Lax pairs are equivalent to each other via a simple transformation. Furthermore, the third-order equation can reduce to a particular ${\rm P_{\uppercase\expandafter{\romannumeral5}}}$ equation with a simple transformation. The limiting kernel can degenerates to two different order of Bessel kernels as $s\rw 0,$ $s\rw \infty,$ respectively. For more and concrete details about the behaviors of the kernel at the hard edge of the equilibrium density, see section 1.2 therein.

\subsection{Main results}

\par
Our main results have intimate relations with a particular solution $r(s)$ of the third-order nonlinear differential equation which can be derived from the compatibility condition of the corresponding Lax pair in \cite{CFC2017},
\bq\label{B1}
8s^2r'(2r'+1)r'''-4s^2(1+4r')r''^2+8sr'(2r'+1)r''-4sr'^2(2r'+1)^2+\lambda^2(2r'+1)^2-4\alpha^2r'^{2}=0,
\eq
where $X'$ denotes $dX/ds,$ $\al>0,$ $\al+\la+1>0,$ and $r(s)$ satisfies the following boundary conditions,
\bq\label{B2}
r(0)=\frac{1-4(\alpha+\lambda)^{2}}{8},\;\;{\rm and}\;\; r(s)=-\la{s^{-\frac{1}{2}}}+\mathcal{O}\left(s^{-1}\right),\;\;s\rw \infty.
\eq
If $\al+\la>0,$
\bq\label{B3}
r'(0)=-\frac{\la}{2(\al+\la)},\;\; s\rw 0.
\eq
The Lax pair of the above third-order nonlinear differential equation is given by the Proposition 2 in \cite{CFC2017}, and the above equation is integrable, see the Remark 2 in \cite{CFC2017}. As pointed out by the
Proposition 4 in \cite{CFC2017}, if
\bq\label{B4}
r'(s)=\frac{y(s)}{2(1-y(s))},
\eq
then $y(s)$ satisfies the following ${\rm P_{\uppercase\expandafter{\romannumeral5}}}$ equation,
\bq\label{B5}
y''(s)=\left(\frac{1}{2y(s)}+\frac{1}{y(s)-1}\right)y'^{2}(s)-\frac{y'(s)}{s}+\frac{(y(s)-1)^{2}}{2s^{2}}\left(\alpha^{2}y(s)-\frac{\lambda^{2}}{y(s)}\right)+\frac{y(s)}{2s},
\eq
where $\al>0,$ $\al+\la+1>0,$ and $y(s)$ satisfies the following boundary conditions,
\bq\label{B6}
y(0)=-\frac{\la}{\al},\;\; {\rm and}\;\; y(s)=-\la{s^{-\frac{1}{2}}}+\mathcal{O}\left(s^{-1}\right),\;\; s\rw \infty.
\eq
\par
The first result is on the asymptotic expansions of Hankel determinants $D_{n}(t; \al, \la)$ are generated from the weight function (\ref{B01}) by sending $n\rw \infty,$ and $t\rw 0,$ such that $s=4nt$ is finite. This expansion in terms of the third-order equation (\ref{B1}), and is uniformly valid for $0<t\leq c,$ where $c$ is a finite positive constant.
\begin{theorem}
Hankel determinants $D_{n}(t; \al, \la)$ associated with the weight function $w(x)$ in (\ref{B01}). Let $n\rw \infty,$ $t \rw 0,$ such that $s=4nt,$ then $D_{n}(t; \al, \la)$ have the following asymptotic expansion,
\bq\label{B7}
D_{n}(t; \al, \la)=D_{n}(0; \al, \la)\exp\left[\int_{0}^{t}\frac{1-4(\al+\la)^{2}-8r(4n\tau)}{16\tau}d\tau\left(1+\mathcal{O}(n^{-\frac{1}{2}})\right)\right],
\eq
where the error term is uniformly valid for $t\in (0, c]$ and $c$ is a finite positive constant, $r(s)$ is analytic for $s\in (0, +\infty)$ and satisfies the third-order equation (\ref{B1}) with the boundary conditions in (\ref{B2}). Moreover, $D_{n}(0; \al, \la)$ have a closed form expression in terms of the Barnes $G-$function, see \cite{M2004},
\bq\label{B8}
D_{n}(0; \al, \la)=\frac{G(n+1)G(n+\al+\la+1)}{G(\al+\la+1)},
\eq
$G(z)$ satisfies a recurrence relation $G(z+1)=\Gamma(z)G(z)$ and the initial condition $G(1)=1,$ see \cite{OLC2010, Vors1987} for a description of $G(z).$
\end{theorem}
\par
For finite $n,$ Chen and McKay have showed that the recurrence coefficient $\al_{n}(t)$ in (\ref{B11}) can be expressed by a solution of a particular Painlev\'e \uppercase\expandafter{\romannumeral5} equation in \cite{YangMcky2012}.
Furthermore, with the aid of the Ladder operator technique, they proved that the logarithmic derivatives of Hankel determinants satisfy a particular Painlev\'e \uppercase\expandafter{\romannumeral5} continuous
Jimbo-Miwa-Okamoto \cite{JM1981, K1981} $\sigma-$form equation. Although the restriction of $\al+\la+1>0$ appears in (\ref{B01}), the conclusions obtained by Chen and McKay in \cite{YangMcky2012} are also valid for this situation. To proceed further, we collect their results in the following theorem.
\begin{theorem}(Chen and McKay \cite{YangMcky2012}). Let $w(x,t)$ be defined in (\ref{B01}), and $R_{n}(t),$ $H_{n}(t)$ be given as follows
\bq\label{B12}
R_{n}(t):=\frac{\la}{h_{n}}\int_{0}^{\infty}\frac{\pi_{n}^{2}(x)w(x,t)}{x+t}dx,\;\;{\rm and}\;\;H_{n}(t):=t\frac{d}{dt}\log{D_{n}(t; \al, \la)}.
\eq
Then, the recurrence coefficients $\al_{n}(t)$ and $\beta_{n}(t)$ can be expressed in terms of $R_{n}(t),$ $H_{n}(t)$ and its derivative as follows
\bq\label{B34}
\al_{n}(t)=2n+1+\al+\la-tR_{n}(t),
\eq
and
\bq\label{B35}
\beta_{n}(t)=n\left(n+\al+\la\right)+tH_{n}'(t)-H_{n}(t).
\eq
\par
Moreover, let
\bq\label{B0001}
y_{n}(t)=\frac{R_{n}(t)}{R_{n}(t)-1},
\eq
then $y_{n}(t)$ satisfies the following ${\rm P_{\uppercase\expandafter{\romannumeral5}}}$ equation,
\bq\label{B13}
y_{n}''=\frac{3y_{n}-1}{2y_{n}(y_{n}-1)}\left(y_{n}'\right)^{2}-\frac{y_{n}'}{t}+\frac{\left(y_{n}-1\right)^{2}}{2t^{2}}\left(\al^{2}y_{n}-\frac{\la^{2}}{y_{n}}\right)+\frac{\left(2n+1+\al+\la\right)y_{n}}{t}-\frac{y_{n}(y_{n}+1)}{2(y_{n}-1)},
\eq
where $X'$ denotes $dX/dt,$with the initial condition $y_{n}(0)=-\frac{\la}{\al}.$
\par
The logarithmic derivatives of Hankel determinants $H_{n}(t)$ satisfy a particular Jimbo-Miwa-Okamoto $\sigma-$form ${\rm P_{\uppercase\expandafter{\romannumeral5}}}$ equation,
\bq\label{B14}
\left(tH_{n}''\right)^{2}=\left[tH_{n}'+(n+\delta_{n})H_{n}'-H_{n}+n\la\right]^{2}-4H_{n}'(H_{n}'+\la)\left(tH_{n}'-H_{n}+n\delta_{n}\right),
\eq
where $\delta_{n}=n+\al+\la,$ $H_{n}(t)$ satisfy the boundary conditions $H_{n}(0)=0.$
\end{theorem}
\begin{proof}
Both definitions of $R_{n}(t)$ and $H_{n}(t)$ are given by (221) and (242) in \cite{YangMcky2012}. The above equation (\ref{B34}), (\ref{B35}) and (\ref{B13}) are equations of (230), (244) and (69) in \cite{YangMcky2012}, respectively. A combination of (224), (225), (246) in \cite{YangMcky2012} and (\ref{B0001}), then it follows (\ref{B13}), see also \cite{CBC2016}. Although the boundary conditions of $y_{n}(t)$ and $H_{n}(t)$ are not written down in \cite{YangMcky2012}, one derived $R_{n}(0)=\la/(\al+\la)$ from the classical Laguerre polynomial in \cite{Szego1939} and (\ref{B12}) with $t=0$ therein.
\end{proof}

\par
We consider the uniform asymptotic approximations of the recurrence coefficient $\al_{n}(t)$ and $\beta_{n}(t)$ for the corresponding monic orthogonal polynomials $\pi_{n}(x)$ in (\ref{B11}) and the leading coefficients $\gamma_{n}(t)$ of the orthonormal polynomials $P_{n}(x)=\gamma_{n}(t)\pi_{n}(x).$ From a pair of Painlev\'e equations (\ref{B13}) and (\ref{B14}) for finite $n,$ with the double scaling scheme, $t\rw 0,$ $n\rw \infty,$ such that $s=4nt$ is finite, we obtain two equations associated with the particular ${\rm P_{\uppercase\expandafter{\romannumeral5}}}$ equation in (\ref{B5}). These results are given in the following Theorem.
\begin{theorem}
Let the weight function $w(x,t)$ be defined in (\ref{B01}), then the recurrence coefficients $\al_{n}(t),$ $\beta_{n}(t)$ for the corresponding monic orthogonal polynomials in (\ref{B11}), and the leading coefficient $\gamma_{n}(t)$ of the orthonormal polynomials $P_{n}(x)=\gamma_{n}(t)\pi_{n}(x)$ have the following uniform asymptotic approximations,
\bq\label{B15}
R_{n}(t)=-2r'(s)\left(1+\mathcal{O}(n^{-\frac{1}{2}})\right),
\eq
\bq\label{B16}
H_{n}(t)=\left(-\frac{r(s)}{2}+\frac{1-4(\al+\la)^{2}}{16}\right)\left(1+\mathcal{O}(n^{-\frac{1}{2}})\right),
\eq
\bq\label{B17}
\al_{n}(t)=2n+1+\al+\la+\frac{sr'(s)}{2n}\left(1+\mathcal{O}(n^{-\frac{1}{2}})\right),
\eq
\bq\label{B18}
\beta_{n}(t)=n(n+\al+\la)+\frac{8r(s)+4(\al+\la)^{2}+1}{16}\left(1+\mathcal{O}(n^{-\frac{1}{2}})\right),
\eq
and
\bq\label{B19}
\gamma_{n}(t)=\frac{1}{\sqrt{\Gamma(n+1)\Gamma(n+1+\al+\la)}}\left(1+\frac{8r(s)+4(\al+\la)^{2}+1}{32}\left(1+\mathcal{O}\left(n^{-\frac{1}{2}}\right)\right)\right),
\eq
where $s=4nt$ and the above error terms are uniformly valid for $t\in (0, c],$ $c$ is a finite positive constant. With the aid of the above approximations, as $n\rw \infty,$
then the Painlev\'e equations (\ref{B13}) and (\ref{B14}) give the following Painlev\'e-type equations satisfied by $r(s),$
\bq\label{B20}
8s^2r'(2r'+1)r'''-4s^2(1+4r')r''^2+8sr'(2r'+1)r''-4sr'^2(2r'+1)^2+\lambda^2(2r'+1)^2-4\alpha^2r'^{2}=0,
\eq
and
\bq\label{B21}
8s^2r''^{2}-16sr'^{3}+\left(16r-8s-2\right)r'^{2}+\left(8r+4\al^{2}-4\la^{2}-1\right)r'-2\la^{2}=0,
\eq
where $X'$ represents $dX/ds.$
\end{theorem}

\begin{rem}
The third order equation is an integrable equation. If one integrates the equation (\ref{B20}) with respect to $s$ on both sides, then it follows (\ref{B21}), see Remark 2 in \cite{CFC2017}. It is also pointed out in \cite{CFC2017}, if one changes the variable as $r'(s)=\frac{y(s)}{2(1-y(s))},$ then the equation (\ref{B20}) turns out to be the ${\rm P_{\uppercase\expandafter{\romannumeral5}}}$ equation in (\ref{B5}).
\par
It is worth note that the third order integrable equation in (\ref{B1}) is the same as the equation in (\ref{B20}), but they are derived from totally different ways. In the present situation the third-order equation in (\ref{B20}) follows from the finite $n$ ${\rm P_{\uppercase\expandafter{\romannumeral5}}}$ equation in (\ref{B13}), and the equation in (\ref{B1}) is a result of the compatibility condition of the corresponding Lax pair in \cite{CFC2017}.
\end{rem}

\par
Under the double scaling scheme, $n\rw \infty,$  $t\rw 0,$ such that $s=4nt,$ the above Theorem 1 and Theorem 3 give the asymptotic approximations in terms of a solution $r(s)$ of the third-order equation associated with the ${\rm P_{\uppercase\expandafter{\romannumeral5}}}$ equation in (\ref{B5}) and its derivative, which have uniform error terms for $t\in (0,c]$ and $n\rw \infty.$ We consider the uniform asymptotic approximations of the recurrence coefficients $\al_{n}(t),$ $\beta_{n}(t)$ and the logarithmic derivatives of the Hankel determinants $H_{n}(t)$ for small $s$ and large $s.$

\begin{cor}
Let $s=4nt,$ $t\in (0,c],$ $c$ is a finite positive constant, the quantities $\al_{n}(t),$ $\beta_{n}(t)$ and $H_{n}(t)$ are given in the Theorem 2, then the uniform asymptotic expansions of $\al_{n}(t),$ $\beta_{n}(t)$ and $H_{n}(t)$ in the following for small $s$ and large $s.$
\par
If $n\rw \infty,$ $t\rw 0^{+},$ such that $s\rw 0^{+}$ and $\al+\la>0,$ then
\bq\label{B24}
H_{n}(t)=\frac{\la{s}}{4(\al+\la)}+\mathcal{O}\left(s^{2}\right)+\mathcal{O}\left(n^{-\frac{1}{2}}\right),
\eq
\bq\label{B23}
\beta_{n}(t)=n(n+\al+\la)+\mathcal{O}(s^{2})+\mathcal{O}\left(n^{-\frac{1}{2}}\right),
\eq
and
\bq\label{B22}
\al_{n}(t)=2n+1+\al+\la-\frac{\la{s}}{4(\al+\la)n}+\mathcal{O}\left(\frac{s^{2}}{n}\right)+\mathcal{O}\left(\frac{1}{n^{\frac{3}{2}}}\right).
\eq
\par
If $n \rw \infty,$ $t\rw c^{-},$ such that $s \rw \infty,$ then
\bq\label{B27}
H_{n}(t)=\frac{1-4(\la+\la)^{2}}{16}\left(1+\mathcal{O}\left(n^{-\frac{1}{2}}\right)\right)+\frac{\la}{2}s^{-\frac{1}{2}}\left(1+\mathcal{O}\left(s^{-\frac{1}{2}}\right)+\mathcal{O}\left(n^{-\frac{1}{2}}\right)\right),
\eq
\bq\label{B25}
\al_{n}(t)=2n+1+\al+\la+\frac{\la}{4n}s^{-\frac{1}{2}}\left(1+\mathcal{O}\left(s^{-\frac{1}{2}}\right)+\mathcal{O}\left(n^{-\frac{1}{2}}\right)\right),
\eq
and
\bq\label{B26}
\beta_{n}(t)=n(n+\al+\la)+\frac{4(\al+\la)^{2}-1}{16}\left(1+\mathcal{O}\left(n^{-\frac{1}{2}}\right)\right)-\frac{3\la}{4}\left(1+\mathcal{O}\left(n^{-\frac{1}{2}}\right)+\mathcal{O}\left(s^{-\frac{3}{2}}\right)\right).
\eq

\end{cor}
\begin{proof}
Combining the asymptotic approximations (\ref{B15})-(\ref{B19}) with the boundary conditions (\ref{B2}) and (\ref{B3}), then it follows (\ref{B24})-(\ref{B26})
\end{proof}
\par
The rest of this paper is organized as follows. In Sect.2, for preparation, we prove that the auxiliary function $R_{n}(t),$ the recurrence coefficients of the monic orthogonal polynomials $\al_{n}(t),$ $\beta_{n}(t)$ and $H_{n}(t)$ can be expressed by a solution of a Riemann-Hilbert problem for the corresponding orthogonal polynomials. In Sect.3, we give the outline and formulas of the Riemann-Hilbert analysis which will benefit for our proof. In Sect.4, we prove the Theorem 1. The proof of the Theorem 3 is presented in Sect.5.

\section{Riemann-Hilbert problem for orthogonal polynomials}
\par
It starts with the following $2\times2$ Riemann-Hilbert problem for $Y(z),$ which associates with the orthogonal polynomials with respect to the weight function $w(x,t)$ in (\ref{B01}).
\par
$(a)$ $Y(z)$ is analytic for $z\in \C \setminus [0,\infty).$
\par
$(b)$ $Y(z)$ satisfies the jump condition
\begin{equation*}
Y_{+}(x)=Y_{-}(x)\left(
\begin{matrix}
1&w(x,t)\\
0&1\\
\end{matrix}
\right),\quad {\rm for}\quad x\in (0,+\infty),
\end{equation*}
where $w(x, t)$ is in (\ref{B01}).
\par
$(c)$ The asymptotic behavior of $Y(z)$ at infinity is
\begin{equation*}
Y(z)=\left(I+\mathcal{O}\left(\frac{1}{z}\right)\right)\left(
\begin{matrix}
z^{n}&0\\
0&z^{-n}\\
\end{matrix}
\right),\quad z\rightarrow\infty.
\end{equation*}
\par
$(d)$ The asymptotic behavior of $Y(z)$ at the origin is
\begin{equation*}
Y(z)=\left(
\begin{matrix}
O(1)&O(1)\\
O(1)&O(1)\\
\end{matrix}
\right)
,\quad z\rightarrow 0.
\end{equation*}
\par
With the general result of Fokas, Its and Kitaev \cite{FIK1992}, in this case, the unique solution of a $2\times2$ matrix valued function $Y(z)$ can be expressed by the corresponding monic orthogonal polynomials as follows
\begin{equation}\label{B28}
Y(z)=\left(
\begin{matrix}
\pi_{n}(z)&\frac{1}{2\pi{i}}\int_{0}^{\infty}\frac{\pi_{n}(x)w(x, t)}{x-z}dx\\
-2\pi{i}\gamma_{n-1}^{2}\pi_{n-1}(z)&-\gamma_{n-1}^{2}\int_{0}^{\infty}\frac{\pi_{n-1}(x)w(x, t)}{x-z}dx\\
\end{matrix}
\right),
\end{equation}
where $\gamma_{n}$ dependents on $t,$ and it is the leading coefficient of the orthogonal polynomial $P_{n}(z)=\gamma_{n}(t)\pi_{n}(z)$ associated with the weight function $w(x,t)$ in (\ref{B01}), and $\pi_{n}(z)$ is the corresponding monic orthogonal polynomials.
\par
For later convenience, we prove that the auxiliary function $R_{n}(t)$ and $H_{n}(t)$ can be expressed by the elements of the above $2\times2-$matrix-valued function $Y(z).$ We list these relation in the following Lemma.

\begin{lemma}
For finite $n,$ the auxiliary function $R_{n}(t)$ in (\ref{B12}), the recurrence coefficient $\al_{n}(t)$ and $H_{n}(t)$ have the following expressions,
\bq\label{B29}
R_{n}(t)=1-i2\pi\al{[Y(0)]_{11}[Y(0)]_{12}}\left(\gamma_{n}(t)\right)^{2},
\eq
\bq\label{B30}
\al_{n}(t)=2n+1+\al+\la-t+i2\pi\al{t}[Y(0)]_{11}[Y(0)]_{12},
\eq
and
\bq\label{B31}
H_{n}'(t)=n-\al{[Y(0)]_{12}[Y(0)]_{21}},
\eq
where $[Y(z)]_{k\ell}$ is the $(k,\ell)$ entry of the matrix function $Y(z)$ in (\ref{B28}).
\par
Moreover, the leading coefficient $\gamma_{n}(t)$ satisfies the equation as follows
\bq\label{B32}
\frac{d}{dt}\ln{\gamma_{n}(t)}=-\frac{1}{2}+i\pi\al{[Y(0)]_{11}[Y(0)]_{12}}\left(\gamma_{n}(t)\right)^{2}.
\eq
\end{lemma}
\begin{proof}
By the definition $R_{n}(t)$ in (\ref{B12}), see also (221) in Chen and McKay \cite{YangMcky2012}, applying integration by parts, one finds,
\begin{align}\label{B33}
R_{n}(t)&=1-\al\left(\gamma_{n}(t)\right)^{2}\int_{0}^{\infty}\frac{\left(\pi_{n}(x)\right)^{2}w(x,t)}{x}dx \nonumber \\
&=1-\al\left(\gamma_{n}(t)\right)^{2}\int_{0}^{\infty}\left(Q_{n-1}(x)+\frac{\pi_{n}(0)}{x}\right)\pi_{n}(x)w(x,t)dx \nonumber\\
&=1-\al\left(\gamma_{n}(t)\right)^{2}\pi_{n}(0)\int_{0}^{\infty}\frac{\pi_{n}(x)w(x,t)}{x}dx,
\end{align}
where $\gamma_{n}(t)$ is the leading coefficient, $Q_{n-1}(x)$ is a combination of the monic orthogonal polynomials $\pi_{k}(x),$ $0\leq k\leq n-1.$ With the $[Y(z)]_{k\ell}$ is the $(k,\ell)$ entry in (\ref{B28}), then (\ref{B33}) follows (\ref{B29}). Obviously, both the equations of (\ref{B34}) and (\ref{B29}) give the equation (\ref{B30}).
\par
Applying the equations (222) and (245) in \cite{YangMcky2012}, and integration by parts, one obtains,
\begin{align}\label{B36}
H_{n}'(t)&=n+\al\left(\gamma_{n-1}(t)\right)^{2}\int_{0}^{\infty}\frac{\pi_{n}(x)\pi_{n-1}(x)w(x,t)}{x}dx\nonumber\\
&=n+\al\left(\gamma_{n-1}(t)\right)^{2}\pi_{n-1}(0)\int_{0}^{\infty}\frac{\pi_{n}(x)w(x,t)}{x}dx,
\end{align}
with the entries of the matrix function $Y(z)$ in (\ref{B28}), (\ref{B36}) gives (\ref{B31}).
\par
Taking the derivative of both sides of (\ref{B05}) with respect to $t,$ and the fact $h_{n}(t)=\left(\gamma_{n}(t)\right)^{-2},$ it implies,
\bq\label{B37}
\frac{\gamma_{n}'(t)}{\gamma_{n}(t)}=-\frac{1}{2}+\frac{\al}{2}\left(\gamma_{n}(t)\right)^{2}\int_{0}^{\infty}\frac{\left(\pi_{n}(x)\right)^{2}w(x,t)}{x}dx,
\eq
by the similar technique in (\ref{B33}) and the entries of $Y(z)$ in (\ref{B28}), it follows the equation (\ref{B32}).
\end{proof}

\section{Deift-Zhou steepest descent analysis}

Following the standard process of the nonlinear steepest descent analysis developed in \cite{DKV1999, DX1999}, we start with the RH problem for $Y$ and obtain the RH problem for $R$ with jump matrices tends to the identity matrix, by a series of transformations $Y\rightarrow T\rightarrow S\rightarrow R.$ Taking a list of inverse transformations, one can derives the uniform asymptotic approximation of the orthogonal polynomials in the complex plane. We are interested in the construction of the local parametrix in the neighborhood of the point $x=0$ of the equilibrium density, in which a particular ${\rm P_{\uppercase\expandafter{\romannumeral5}}}$ equation is involved. This ${\rm P_{\uppercase\expandafter{\romannumeral5}}}$ equation is equivalent to the general ${\rm P_{\uppercase\expandafter{\romannumeral3}}}$ equation \cite{GLS2002}, see also \cite{CBC2016} for the concrete model of the single-user Multi-Input-Multi-Output (MIMO) wireless communication systems.
\par
In this section, we review the main steps and formulas to derive the uniformly asymptotic approximations of the Hankel determinants and the recurrence coefficients $\al_{n}(t)$ and $\beta_{n}(t)$ of the corresponding monic orthogonal polynomials.
\\
{\bf Utilization and the first transformation: $Y\rightarrow T.$}
\par
To normalize the matrix function $Y(z)$ at infinity, the first transformation is given by
\bq\label{B37}
T(z)=(4n)^{-(n+\frac{\alpha+\lambda}{2})\sigma_{3}}e^{-\frac{n\ell}{2}\sigma_{3}}Y(4nz)e^{-n(g(z)-\frac{\ell}{2})\sigma_{3}}
\left(4n\right)^{\frac{\alpha+\lambda}{2}\sigma_{3}},\; \; \; z\in \mathbb{C}\setminus [0,+\infty),
\eq
where $g(z)$ is defined in terms of the equilibrium density $\mu(x)=\frac{2}{\pi}\sqrt{\frac{1-x}{x}},$ $x\in (0,1)$ as follows
\bq\label{B38}
g(z)=\int_{0}^{1}\log(z-x)\mu(x)dx,
\eq
and $\ell=-2(1+2\ln{4})$ is the Euler-Laguerre constant.
\par
Actually, by the above transformation, $T(z)$ is utilized such that $T(z)=I+ \mathcal{O}\left(1/z\right),$ as $z\rightarrow \infty.$
\par
For statements convenience, we claim that the following Pauli matrices$\sigma_{k},$ $k=1, 2, 3$ have been used in this paper,
\bbq
\sigma_{1}=\left(
\begin{matrix}
0&1\\
1&0\\
\end{matrix}
\right),
\;\;
\sigma_{2}=\left(
\begin{matrix}
0&-i\\
i&0\\
\end{matrix}
\right)
\;\;
\sigma_{3}=\left(
\begin{matrix}
1&0\\
0&-1\\
\end{matrix}
\right).
\eeq
\begin{center}
\begin{tikzpicture}
\begin{scope}[line width=2pt]
\draw[->,>=stealth] (0,0)--(3,0);

\draw[-] (0,0)--(10,0);
\draw[->,>=stealth] (6,0)--(8,0);
\draw[->,>=stealth] (2.9,1.87)--(3.1,1.87);
\draw[->,>=stealth] (2.9,-1.87)--(3.1,-1.87);
\draw[-] (3,0)--(5,0);
\draw (0,0) .. controls (2.,2.5) and (4.0,2.5) .. (6,0);
\draw (0,0) .. controls (2.,-2.5) and (4.0,-2.5) .. (6,0);
\node[below] at (0,0) {$O$};
\node[above] at (1.6,1.6) {$\Sigma_{3}$};
\node[above] at (1.8,0) {$\Sigma_{2}$};
\node[above] at (1.8,-1.5) {$\Sigma_{1}$};
\node[below] at (6,0) {$1$};
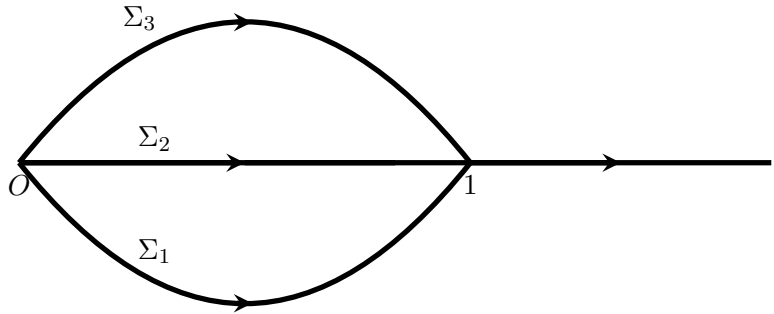
\node[below] at (3.5,-2.9) { Figure 1. The contour $\Sigma_{S}=\{\bigcup_{k=1}^{3}\Sigma_{k}\}\bigcup(1,\infty)$ for $S(z).$};
\end{scope}
\end{tikzpicture}
\end{center}

{\bf Opening of the lens on $(0,1)$}
\par
From the expression of $g(z)$ in (\ref{B38}), it is not analytic on the equilibrium support $(0, 1).$ With the equilibrium density technique in \cite{DeiftBook} and the properties of $g(z),$ it is need to introduce the following function
\bq\label{BB39}
\phi(z)=2\int_{0}^{z}\sqrt{\frac{s-1}{s}}ds,\quad z\in \C \setminus [0,+\infty),
\eq
where $\arg{z}\in (0,2\pi),$ $\phi(z)$ satisfies that $\phi(z)$ is purely imaginary for $x\in (0,1);$ $\phi(x)>0$ for $x>1;$ $\Re\phi(z)<0$ for $\Im{z}\neq 0$ and $0<\Re{z}<1.$
\par
After some calculations, see also \cite{CFC2017}, $T(z)$ satisfies the following jump condition
\bq\label{B39}
T_{+}(x)=T_{-}(x)\left\{
\begin{array}{ll}
\left(
\begin{matrix}
e^{2n\phi_{+}(x)}&x^{\alpha}\left(x+\frac{t}{4n}\right)^{\lambda}\\
0&e^{2n\phi_{-}(x)}\\
\end{matrix}
\right), & x\in (0,1),
\\
\left(
\begin{matrix}
1&x^{\alpha}\left(x+\frac{t}{4n}\right)^{\lambda}e^{-2n\phi(x)}\\
0&1\\
\end{matrix}
\right),& x\in \left(1,+\infty\right).
\end{array}
\right.
\eq
\par
The jump matrix of $T$ on the interval $(0,\infty)$ tends to the identity matrix very quickly as $n\rightarrow \infty,$ but the situation on the interval $(0,1)$ is more complicated due to the oscillation entries therein. By the factorization of the jump matrix, the oscillation entries are translated to the off-diagonal terms in the jump matrix and the interval $(0,1)$ is deformed as an opening lens, this is the key step of the Deift-Zhou nonlinear steep descent method.
\par
The second transformation is introduced as follows
\bq\label{B40}
S(z)=\left\{
\begin{array}{lll}
T(z),&{\rm for}\; z\; {\rm outside}\; {\rm the}\; {\rm lens}\; {\rm shaped}\; {\rm region},\\
\\
T(z)\left(
\begin{matrix}
1&0\\
-z^{-\alpha}\left(z+\frac{t}{4n}\right)^{-\lambda}e^{2n\phi(z)}&1\\
\end{matrix}
\right),&{\rm for}\; z\; {\rm in}\; {\rm the}\; {\rm upper}\; {\rm lens}\; {\rm region},\\
\\
T(z)\left(
\begin{matrix}
1&0\\
z^{-\alpha}\left(z+\frac{t}{4n}\right)^{-\lambda}e^{2n\phi(z)}&1\\
\end{matrix}
\right),& {\rm for}\; z\; {\rm in}\; {\rm the}\; {\rm lower}\; {\rm lens}\; {\rm region},
\end{array}
\right.
\eq
where $\arg{z}\in (-\pi, \pi).$ Obviously, $S(z)$ keeps the utilization, namely, $S(z)=I+\mathcal{O}\left(1/z\right)$ as $z\rightarrow \infty,$ and $S(z)$ satisfies $S_{+}(z)=S_{-}(z)J_{S}$ and $J_{S}$ is given as follows
\bq\label{B41}
J_{S}(z)=\left\{
\begin{array}{lll}
\left(
\begin{matrix}
1&0\\
z^{-\alpha}\left(z+\frac{t}{4n}\right)^{-\lambda}e^{2n\phi(z)}&1\\
\end{matrix}
\right),&{\rm for}\; z \in \Sigma_{1}\bigcup\Sigma_{3},\\
\\
\left(
\begin{matrix}
0&x^{\alpha}\left(x+\frac{t}{4n}\right)^{\lambda}\\
-x^{-\alpha}\left(x+\frac{t}{4n}\right)^{-\lambda}&0\\
\end{matrix}
\right),& {\rm for}\;z=x\in \Sigma_{2}=(0,1),\\
\\
\left(
\begin{matrix}
1&x^{\alpha}\left(x+\frac{t}{4n}\right)^{\lambda}e^{2n\phi(x)}\\
0&1\\
\end{matrix}
\right),& {\rm for}\; z=x\in (1,+\infty).
\end{array}
\right.
\eq
By the second transformation, there is no oscillation term in the diagonal of the jump matrix of $S(z).$

\subsection{Global parametrix}

From the jump matrix $J_{S}$ in (\ref{B41}), it tends to the identity matrix for $z\in \Sigma_{1}\bigcup\Sigma_{3}\bigcup(1,+\infty)$ as $n\rightarrow \infty,$ and $\left(x+t/4n\right)^{\pm\lambda}=x^{\pm\lambda}+\mathcal{O}(n^{-1})$ as $n\rightarrow \infty,$ are valid for $x\in (0, 1)$ and $t\in (0,c],$ where $c$ is a finite and positive constant. Then, it turns out to be seeking a solution of the RH problem for $P^{(\infty)}(z)$ as follows
\par
$(a)$ $P^{(\infty)}(z)$ is analytic in $\C \setminus [0,1].$
\par
$(b)$ $P^{(\infty)}(z)$ satisfies the jump condition,
\bq\label{B42}
P^{(\infty)}_{+}(x)=P^{(\infty)}_{-}(x)
\left(
\begin{matrix}
0&x^{\al+\la}\\
-x^{-(\al+\la)}&0\\
\end{matrix}
\right),\;\; {\rm for}\;\;x \in \Sigma_{2}=(0,1).\\
\eq
\par
$(c)$ The asymptotic behavior at infinity is
\bbq
P^{(\infty)}(z)=I+\mathcal{O}\left(z^{-1}\right),\;\; z\rightarrow \infty.
\eeq
\par
A solution of the above RH problem, $P^{(\infty)}(z)$ can be constructed as follows
\bq\label{B43}
P^{(\infty)}(z)=D(\infty)^{\sigma_{3}}M^{-1}a^{-\sigma_{3}}(z)MD^{\sigma_{3}}(z), \;\; {\rm for}\;\; z\in \mathbb{C}\setminus [0,1],
\eq
where $M=(I+i\sigma_{1})/\sqrt{2},$ $a(z)=((z-1)/z)^{\frac{1}{4}},$ $D(\infty)=2^{-(\al+\la)},$ and $D(z)=(\sqrt{z}/(\sqrt{z}+\sqrt{z-1}))^{-(\al+\la)}$ is the Szeg\"{o} function, such that $D_{+}(x)D_{-}(x)=x^{\al+\la}$ for $x\in (0,1),$ moreover, the branches are taken as $\arg{z}\in (-\pi, \pi)$ and $\arg(z-1)\in (-\pi, \pi).$
\par
Then,
\bq\label{B001}
S(z){P^{(\infty)}}^{-1}(z)=I+\mathcal{O}\left(n^{-1}\right),\;\;n\rw\infty,
\eq
where the error term is uniform for $z$ away from the end points $0,$ $1$ and $t\in(0,c],$ $c$ is a positive finite constant.
Due to the factor $a(z)$ in $(\ref{B43})$ with fourth root singularities at $z=0$ and $1,$ so it is need to construct the local parametrices in the neighborhoods of these points.
\par
Let
\bq\label{B44}
N(z)=P^{(\infty)}(z)\left(\frac{z+\frac{t}{4n}}{z}\right)^{-\frac{\la}{2}\sigma_{3}},
\eq
where $\arg z \in (-\pi, \pi),$ $\arg(z+t/4n) \in (-\pi,\pi),$ and $P^{(\infty)}(z)$ is given by (\ref{B43}). Then, $N(z)$ satisfies the RH problem for $P^{(\infty)}(z)$ and the jump condition $(\ref{B42})$ replaced by
\bq\label{B45}
N_{+}(x)=N_{-}(x)\left\{
\begin{array}{lll}
e^{i\la\pi\sigma_{3}},& \; x\in (-\frac{t}{4n}, 0),\\
\\
\left(
\begin{matrix}
0&x^{\alpha}\left(x+\frac{t}{4n}\right)^{\lambda}\\
-x^{-\alpha}\left(x+\frac{t}{4n}\right)^{-\lambda}&0\\
\end{matrix}
\right),& \; x\in (0,1).\\
\end{array}
\right.
\eq
Note that, $N(z)$ has the same jump condition with $S(z)$ for fixed $z=x\in (0,1).$

\subsection{Local parametrix $P^{(1)}(z)$ at $z=1$}

For $z\in U(1,r),$ $U(1,r)=\{z:|z-1|<r\},$ and $r$ is sufficiently small and positive, then the local parametric $P^{(1)}(z)$ can be expressed by the Airy function and its derivatives. Especially, for $z \in \partial U(1,r),$ $P^{(1)}(z)$ satisfies the matching condition as follows
\bq\label{B002}
P^{(1)}(z)P^{(\infty)-1}(z)=I+\mathcal{O}\left(n^{-1}\right),\; {\rm as}\;\; n\rw \infty.
\eq
More concrete details showed in \cite{CFC2017}.

\subsection{Local parametrix $P^{(0)}(z)$ at $z=0$}

For $z \in U(0,r)=\{z\in \mathbb{C}: |z|<r\}$ for small $r,$ and $r>t/4n,$ the parametrix $P^{(0)}(z)$ satisfies the following RH problem,
\par
$(a)$ $P^{(0)}(z)$ is analytic in $U(0,r)\setminus \{\bigcup_{k=1}^{3}\Sigma_{k}\}.$
\par
$(b)$ $P^{(0)}(z)$ has the same jump conditions with $S(z)$ on $U(0,r)\bigcap \Sigma_{k},k=1,2,3,$ see (\ref{B41}).
\par
$(c)$ For $z\in \partial U(0,r)=\{z\in \mathbb{C}:|z|=r\},$ $P^{(0)}(z)$ satisfies the following matching condition,
\bq\label{B46}
P^{(0)}(z){P^{(\infty)}}^{-1}(z)=I+\mathcal{O}\left(n^{-\frac{1}{2}}\right), \;\; {\rm as}\;\; n\rightarrow \infty.
\eq
\par
$(d)$ The asymptotic behavior of $P^{(0)}(z)$ at $z=0$ is the same as $S(z).$
\par

In this situation, the local parametrix $P^{(0)}(z)$ is related to a model RH problem for $\Phi(\xi,s).$ From the Lax pair of $\Phi(\xi,s),$ we obtain a third-order non linear differential equation which can degenerate to a particular ${\rm P_{\uppercase\expandafter{\romannumeral5}}}$ equation, the more information see \cite{CFC2017}. The local parametric $P^{(0)}(z)$ can be expressed in terms of $\Phi(\xi,s)$ as follows
\bq\label{B47}
P^{(0)}(z)=E(z)\Phi(n^{2}\phi^{2}(z), 4nt)e^{-\frac{i\pi}{2}\sigma_{3}}(-z)^{-\frac{\alpha+\lambda}{2}\sigma_{3}}\left(\frac{-\left(z+\frac{t}{4n}\right)}{-z}\right)^{-\frac{\lambda}{2}\sigma_{3}}e^{n\phi(z)\sigma_{3}},
\eq
where $z\in U(0,r)\setminus \Sigma_{S},$ and $E(z)$ is given by
\bq\label{B48}
E(z)=N(z)e^{\frac{i\pi}{2}\sigma_{3}}\left(-z\right)^{\frac{\al+\la}{2}\sigma_{3}}\left(\frac{-\left(z+\frac{t}{4n}\right)}{-z}\right)^{\frac{\la}{2}\sigma_{3}}\frac{I-i\sigma_{1}}{\sqrt{2}}\left(n^{2}\phi^{2}(z)\right)^{\frac{1}{4}\sigma_{3}},
\eq
$N(z)$ is in (\ref{B44}), the branches are chosen such that $\arg(-z)\in (-\pi, \pi),$ $\arg\{n^{2}\phi^{2}(z)\}\in (-\pi,\pi),$ $\arg\{-(z+t/4n)\}\in (-\pi, \pi),$ and $E(z)$ is analytic in $U(0,r).$
\par
\begin{center}
\begin{tikzpicture}
\begin{scope}[line width=2pt]
\draw[->,>=stealth] (0,0)--(1.5,0);
\draw[-] (0,0)--(3,0);
\node[below] at (3,0) {$S$};

\draw[->,>=stealth] (-3.5,0)--(-2,0);
\draw[-] (-4.6,0)--(0,0);
\draw[dashed] (3,0)--(6,0);
\draw[->,>=stealth] (-3.2,2.8)--(-1.6,1.4);
\draw[-] (-2.4,2.1)--(0,0);
\draw[->,>=stealth] (-3.2,-2.8)--(-1.6,-1.4);
\draw[-] (-2.4,-2.1)--(0,0);
\node[below] at (0,0) {$O$};
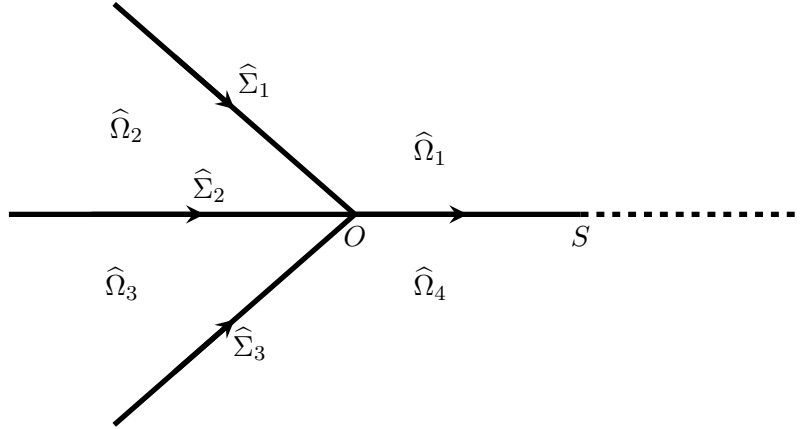
\node[below] at (-.8,-3.5) {Figure 2. Contours $\mathbb{C}\setminus \displaystyle\cup_{j=1}^{3}\widehat{\Sigma}_{j}\cup\left(0,s\right),$ and regions $\widehat{\Omega}_{j}, j=1,\ldots,4.$ };

\node[above] at (-1.4,-2.1) {{$\widehat{\Sigma}_{3}$}};
\node[above] at (-2,0) {{ $\widehat{\Sigma}_{2}$}};
\node[above] at (-1.4,1.4) {{ $\widehat{\Sigma}_{1}$}};

\node[above] at (1,.5) {{$\widehat{\Omega}_{1}$}};
\node[above] at (-3.1,.8) {{ $\widehat{\Omega}_{2}$}};
\node[above] at (-3.1,-1.3) {{$\widehat{\Omega}_{3}$}};
\node[above] at (1,-1.3) {{$\widehat{\Omega}_{4}$}};
\end{scope}
\end{tikzpicture}
\end{center}

\par
A model RH problem for $\Phi(\xi, s)$ is given as follows
\par
$(a)$ $\Phi(\xi,s)$ is analytic in $\mathbb{C}\setminus \displaystyle\cup_{j=1}^{3}\widehat{\Sigma}_{j}\cup\left(0,s\right),$ see Figure 2.
\par
$(b)$ $\Phi(\xi,s)$ satisfies the following jump condition
\bq\label{B49}
\Phi_{+}(\xi,s)=\Phi_{-}(\xi,s)\left\{
\begin{array}{llll}
e^{{i}\lambda\pi{\sigma_{3}}}, & \xi\in \left(0,s\right), \\
\\
\left(
\begin{matrix}
1&0\\
e^{i\pi(\lambda+\alpha)}&1\\
\end{matrix}
\right),& \xi\in \widehat{\Sigma}_{1},\\
\\
\left(
\begin{matrix}
0&1\\
-1&0\\
\end{matrix}
\right),& z \in \widehat{\Sigma}_{2},\\
\\
\left(
\begin{matrix}
1&0\\
e^{-i\pi(\lambda+\alpha)}&1\\
\end{matrix}
\right),& \xi \in\widehat{\Sigma}_{3}.
\end{array}
\right.
\eq
\par
$(c)$ As $\xi\rightarrow \infty,$ the asymptotic behavior of $\Phi(\xi,s)$ is described by
\bq\label{B50}
\Phi(\xi,s)=\left(I+\frac{C_{1}(s)}{\xi}+\mathcal{O}\left(\frac{1}{\xi^{2}}\right)\right)\xi^{-\frac{1}{4}\sigma_{3}}\frac{I+i\sigma_{1}}{\sqrt{2}}e^{\sqrt{\xi}\sigma_{3}},
\eq
where the matrix function $C_{1}(s)$ is independent of $\xi$ and $\arg{\xi} \in (-\pi, \pi).$
\par
$(d)$ As $\xi\rightarrow 0,$ the asymptotic behavior of $\Phi(\xi,s)$ is
\bq\label{B51}
\Phi(\xi,s)=Q_{1}(s)\left(I+\mathcal{O}(\xi)\right)\xi^{\frac{\alpha}{2}\sigma_{3}}\left\{
\begin{array}{llll}
e^{\frac{{i}\lambda\pi}{2}\sigma_{3}}, & \xi\in \widehat{\Omega}_{1}, \\
\\
\left(
\begin{matrix}
1&0\\
-e^{i\pi(\lambda+\alpha)}&1\\
\end{matrix}
\right),& \xi \in \widehat{\Omega}_{2},\\
\\
\left(
\begin{matrix}
1&0\\
e^{-i\pi(\lambda+\alpha)}&1\\
\end{matrix}
\right),& \xi\in \widehat{\Omega}_{3},\\
\\
e^{-\frac{{i}\lambda\pi}{2}\sigma_{3}}, & \xi\in \widehat{\Omega}_{4}, \\
\end{array}
\right.
\eq
where four sectors $\widehat{\Omega}_{j}, j=1,\ldots,4,$ illustrated in Figure 2, $\arg{\xi}\in (-\pi, \pi)$ and the matrix $Q_{1}(s)$ independent of $\xi,$ such that $\det\left(Q_{1}(s)\right)=1.$
\\
$(e)$ As $\xi \rightarrow s,$ the asymptotic behavior $\Phi(\xi,s)$ is given by
\bq\label{B52}
\Phi(\xi,s)=Q_{2}(s)\left(I+\mathcal{O}(\xi-s)\right)\left(\xi-s\right)^{\frac{\lambda}{2}\sigma_{3}},
\eq
where $\arg(\xi-s)\in (-\pi, \pi),$ and the matrix function $Q_{2}(s)$ independent of $\xi,$ such that $\det\left(Q_{2}(s)\right)=1.$

\subsection{Riemann-Hilbert problem for $R$}
\par
The finial transformation $S\rw R$ and $R(z)$ is defined as follows
\bq\label{B53}
R(z)=\left\{
\begin{array}{llll}
S(z)P^{(\infty)-1}(z),& z \in \mathbb{C}\setminus\{U(0,r)\cup U(1,r)\cup \Sigma_{S}\} ,\\
\\
S(z)P^{(0)-1}(z),& z\in U(0, r)\setminus \Sigma_{S},\\
\\
S(z)P^{(1)-1}(z),& z\in U(1, r)\setminus \Sigma_{S}.\\
\end{array}
\right.
\eq
Then, $R(z)$ satisfies the following RH problem.
\par
$(a)$ $R(z)$ is analytic in $\mathbb{C}\setminus \Sigma_{R}$, see Figure 3.
\par
$(b)$ $R(z)$ satisfies $R_{+}(z)=R_{-}(z)J_{R}(z)$ and the jump $J_{R}(z)$ is given by
\bq\label{B54}
J_{R}(z)=\left\{
\begin{array}{llll}
P^{(\infty)}(z)J_{S}(z)P^{(\infty)-1}(z),& z \in \Sigma_{R}\setminus \{\partial U(0,r)\cup \partial U(1,r)\}, \\
\\
P^{(0)}(z)P^{(\infty)-1}(z),& z\in \partial U(0, r),\\
\\
P^{(1)}(z)P^{(\infty)-1}(z),& z\in \partial U(1, r),\\
\end{array}
\right.
\eq
where $J_{S}$ denotes the jump matrices in (\ref{B41}).
\par
$(c)$ $R(z)=I+\mathcal{O}\left(z^{-1}\right)$ as $z\rw \infty.$
\par

\begin{center}
\begin{tikzpicture}
\begin{scope}[line width=2pt]
\draw(0,0) circle(0.8);
\draw(7,0) circle(0.8);
\fill (0,0) circle (1pt);
\fill (7,0) circle (1pt);

\draw[->,>=stealth] (0.8,0)--(3,0);
\draw[-] (2.6,0)--(6.2,0);

\draw[->,>=stealth] (7.8,0)--(9,0);
\draw[->,>=stealth] (3.4,1.6)--(3.5,1.6);
\draw[->,>=stealth] (3.4,-1.6)--(3.5,-1.6);

\draw[-] (7.8,0)--(11,0);
\draw (0.567,0.567) .. controls (3.4,1.9) and (3.6,2) .. (6.433,.567);
\draw (0.567,-0.567) .. controls (3.4,-1.9) and (3.6,-2) .. (6.433,-.567);
\node[below] at (0,0.05) { $0$};

\node[above] at (2.8,1.7) {$\widetilde{\Sigma}_{3}$};
\node[above] at (3.8,0) {$\widetilde{\Sigma}_{2}$};
\node[above] at (2.9,-1.6) {$\widetilde{\Sigma}_{1}$};
\node[below] at (7,0.05) { $1$};
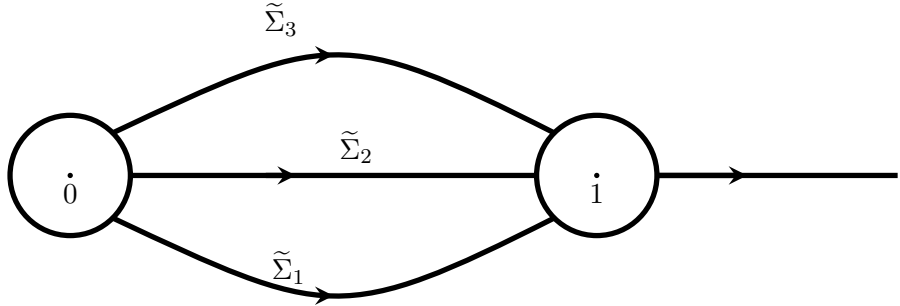
\node[below] at (4.6,-2) { Figure 3. The contour  for $R(z).$};
\end{scope}
\end{tikzpicture}
\end{center}

By the matching conditions on the boundaries (\ref{B001}), (\ref{B002}), (\ref{B46}) and $\phi(z)$ in (\ref{B39}), one finds,
\bq\label{B55}
J_{R}(z)=\left\{
\begin{array}{llll}
I+\mathcal{O}\left(n^{-\frac{1}{2}}\right),& z\in \partial U(0, r)\cup \partial U(1,r),\\
\\
I+\mathcal{O}\left(n^{-1}\right),& z \in \widetilde{\Sigma}_{2},\\
\\
I+\mathcal{O}\left(e^{-cn}\right),& z \in \Sigma_{R}\setminus \{\partial U(0,r)\cup \partial U(1,r)\cup \widetilde{\Sigma}_{2}\},\\
\end{array}
\right.
\eq
where $c$ is a positive constant, and the error term is uniform for $z \in \Sigma_{R}.$ Then one finds,
\bq\label{B56}
\|J_{R}(z)-I\|_{L^{2}\cap L^{\infty}(\Sigma_{R})}=\mathcal{O}\left(n^{-\frac{1}{2}}\right).
\eq
Moreover,
\bq\label{B57}
R(z)=I+\frac{1}{i2\pi}\int_{\Sigma_{R}}\frac{R_{-}(\tau)\left(J_{R}(\tau)-I\right)}{\tau-z}d\tau,\;z\notin \Sigma_{R},
\eq
by the standard method and procedure of norm estimation of Cauchy operator and the approach in \cite{DeiftBook, DX1999}, one derives the following estimation from (\ref{B56}) and (\ref{B57}) that
\bq\label{B58}
R(z)=I+\mathcal{O}\left(n^{-\frac{1}{2}}\right),
\eq
where the error term is uniform for $z\in \mathbb{C},$ and $t\in (0,c]$ and $c$ is a finite positive constant.
\par
It now completes the nonlinear steepest descent analysis from $Y$ to $R.$

\section{Proof of Theorem 1}

Tracing back of the above nonlinear steepest descent analysis, we can derive the asymptotic approximation of the matrix function $Y(z)$ in the whole complex plane, especially the special case $Y(0).$ Note that the auxiliary function $R_{n}(t),$ the recurrence coefficient $\al_{n}(t),$ the logarithmic derivative of the Hankel determinant $H_{n}(t)$ and the logarithmic derivative of the leading coefficient $\gamma_{n}(t)$ can be expressed in terms of the entries of the matrix function $Y(0)$ in the Lemma 1, then we can derive one-term asymptotic approximation has a relation with the ${\rm P_{\uppercase\expandafter{\romannumeral5}}}$ equation in (\ref{B5}).
\\
\par
{\bf Proof of Theorem 1.} Taking the inverse transformations of $Y\rightarrow T\rightarrow S\rightarrow R$ and the exact formulas in (\ref{B37}), (\ref{B40}), (\ref{B47})-(\ref{B48}) and (\ref{B53}), one finds,
\bq\label{BB58}
Y_{+}(4nx)=M_{0}R(x)E(x)\Phi_{-}(n^{2}\phi^{2}(x),4nt)e^{\frac{i\pi(\al+\la-1)}{2}\sigma_{3}}
\left(
\begin{matrix}
1&0\\
1&1\\
\end{matrix}
\right)
\left(w(4nx)\right)^{-\frac{1}{2}\sigma_{3}},
\eq
and
where $M_{0}=(-1)^{n}(4n)^{\left(n+\frac{\al+\la}{2}\right)\sigma_{3}}e^{\frac{n\ell}{2}\sigma_{3}}$ and $x\in (0,r),$ $\Phi(\xi, s)$ satisfies the RH problem (\ref{B49})-(\ref{B52}), $E(z)$ and $R(z)$ are given by (\ref{B48}), (\ref{B58}), respectively.
\par
To find out $Y(0),$ we need to evaluate the asymptotic expansions of $E(x)$ and $\Phi_{-}(n^{2}\phi^{2}(x),4nt)$ at $x=0.$
\par
By the concrete expression in (\ref{BB39}), the Maclaurin expansion $\phi(x)=i4x^{\frac{1}{2}}+\mathcal{O}\left(x^{\frac{3}{2}}\right),$ $x\rightarrow 0$ and $E(z)$ in (\ref{B48}), after some calculations, one finds,
\bq\label{B59}
E(x)=2^{-(\al+\la)\sigma_{3}}2^{-\frac{1}{2}}\left(I-i\sigma_{1}\right)
\left(
\begin{matrix}
0&1\\
-1&\al+\la\\
\end{matrix}
\right)
(4n)^{\frac{1}{2}\sigma_{3}}+\mathcal{O}(x),\;\; x\rightarrow 0.
\eq
With the aid of the asymptotic behavior of $\Phi(\xi, s)$ in (\ref{B51}), one follows that
\begin{align}\label{B60}
\Phi&_{-}(n^{2}\phi^{2}(x),4nt)e^{\frac{i\pi(\al+\la-1)}{2}\sigma_{3}}\left(
\begin{matrix}
1&0\\
1&1\\
\end{matrix}
\right)
\left(w(4nx)\right)^{-\frac{1}{2}\sigma_{3}}\nonumber\\
&=Q_{1}(4nt)\left(I+\mathcal{O}\left(L_{n}(x)\right)\right)\left(L_{n}(x)\right)^{\frac{\al}{2}\sigma_{3}}\left(w(4nx)\right)^{-\frac{1}{2}\sigma_{3}},
\end{align}
where $L_{n}(z)=n^{2}\phi^{2}(z)=\xi$ plays the role of a conformal mapping for $z \in U(0,r),$ $r$ is sufficiently small, $r>\frac{t}{4n},$ $t\in (0, c]$ and $c$ is a finite positive constant.
\par
For convenience, we restate the equations (1.10) and (2.35) in \cite{CFC2017} are satisfied by the auxiliary matrix function $A_{1}(s)$ as follows
\bq\label{B61}
A_{1}(s)=\frac{\alpha}{2}Q_{1}(s)\sigma_{3}Q_{1}^{-1}(s),
\eq
and
\bq\label{B62}
A_{1}(s)=\left(
\begin{matrix}
-\frac{1}{4}+\frac{1}{2}r(s)+q'(s)&-i\left(\frac{1}{2}+r'(s)\right)\\
i\left(-q(s)+t'(s)\right)&\frac{1}{4}-\frac{1}{2}r(s)-q'(s)\\
\end{matrix}
\right),
\eq
where $q(s)$ and $t'(s)$ in terms of $r(s)$ and its derivatives,
\bq\label{B64}
t'(s)=\frac{\lambda^{2}-[2sr''(s)-2r'(s)r(s)+r'(s)]^2}{4r'(s)},
\eq
\bq\label{B65}
q(s)=-\frac{8s^2[r''(s)]^2-r'(s)[4r(s)(r(s)-1)(2r'(s)+1)+2r'(s)+4\lambda^2-4\alpha^2+1]-2\lambda^2}{8r'(s)(2r'(s)+1)},
\eq
with $\al>0,$ $\al+\la+1>0.$
\par
By the equivalence of (\ref{B61}) and (\ref{B62}), and some calculations, one gets,
\bq\label{B63}
Q_{1}(s)=[Q_{1}(s)]_{11}\left(
\begin{matrix}
1&\frac{i(4q'(s)2r(s)-1-2\al)}{4(q(s)-t'(s))}\\
-\frac{i(4q'(s)+2r(s)-1-2\al)}{2(1+2r'(s))}&1\\
\end{matrix}
\right)c_{2}^{\sigma_{3}},
\eq
where $[Q_{1}(s)]_{11}$ is given by
\bq\label{B67}
[Q_{1}(s)]_{11}=\sqrt{\frac{2(1+2r'(s))(q(s)-t'(s))}{\al\left(4q'(s)+2r(s)-1-2\al\right)}},
\eq
and $t'(s),$ $q(s)$ are given by (\ref{B64}), (\ref{B65}), respectively. $c_{2}$ is an arbitrary non-zero constant, the more details can be found in the Section 2 in \cite{CFC2017}.
\par
With the asymptotic expansion of $E(x)$ in (\ref{B59}) and $Q_{1}(s)$ in (\ref{B63}), as $x\rightarrow 0,$
{\small
\bq\label{B66}
[E(x)Q_{1}(s)]_{11}=\frac{in^{\frac{1}{2}}[Q_{1}(s)]_{11}c_{2}}{2^{\al+\la-\frac{1}{2}}}\left[1-\frac{\left(1-i(\al+\la)\right)\left(4q'(s)+2r(s)-1-2\al\right)}{8(1+2r'(s))n}\right]+\mathcal{O}(x),
\eq
}
{\small
\bq\label{B68}
[E(x)Q_{1}(s)]_{12}=\frac{n^{\frac{1}{2}}\left(1-i(\al+\la)\right)[Q_{1}(s)]_{11}}{2^{\al+\la+\frac{3}{2}}c_{2}}\left[-\frac{4q'(s)+2r(s)-1-2\al}{\left(1-i(\al+\la)\right)\left(q(s)-t'(s)\right)}+\frac{1}{n}\right]+\mathcal{O}(x),
\eq
}
and
{\small
\bq\label{B69}
[E(x)Q_{1}(s)]_{21}=-n^{\frac{1}{2}}2^{\al+\la+\frac{1}{2}}[Q_{1}(s)]_{11}c_{2}\left[1+\frac{\left(1+i(\al+\la)\right)\left(4q'(s)+2r(s)-1-2\al\right)}{8\left(1+2r'(s)\right)n}\right]+\mathcal{O}(x),
\eq
}
where $[Q_{1}(s)]_{11}$ is in (\ref{B67}).
\par
By the formulas (\ref{BB58})-(\ref{B60}) and (\ref{B66})-(\ref{B69}), it follows that
\begin{align}\label{B70}
\left[Y(0)\right]_{11}\left[Y(0)\right]_{12}&=(4n)^{2n+\al+\la}e^{n\ell}[E(0)Q_{1}(s)]_{11}[E(0)Q_{1}(s)]_{12}\left(1+\mathcal{O}\left(n^{-\frac{1}{2}}\right)\right)\nonumber\\
&=-\frac{i}{\al}n^{2n+1+\al+\la}4^{2n}e^{n\ell}\left(1+2r'(s)\right)\left(1+\mathcal{O}\left(n^{-\frac{1}{2}}\right)\right),
\end{align}
and
\begin{align}\label{B71}
\left[Y(0)\right]_{12}\left[Y(0)\right]_{21}&=[E(0)Q_{1}(s)]_{12}[E(0)Q_{1}(s)]_{21}\left(1+\left(n^{-\frac{1}{2}}\right)\right)\nonumber\\
&=\frac{n}{\al}\left(1+2r'(s)\right)\left(1+\mathcal{O}\left(n^{-\frac{1}{2}}\right)\right),
\end{align}
where $s=4nt$ and the above error terms are uniform for $t\in (0,c],$ $c$ is a finite positive constant.
\par
From the formulas (\ref{B31}) and (\ref{B71}), one finds,
\bq\label{B72}
H_{n}'(t)=-2nr'(s)\left(1+\mathcal{O}\left(n^{-\frac{1}{2}}\right)\right), \; \; s=4nt.
\eq
\par
Integrating the above equation from $0$ to $t$ on both sides and combining with the boundary conditions $r(0)=\frac{1-4(\al+\la)^{2}}{8}$ and $H_{n}(0)=0,$ then one finds,
\bq\label{B73}
H_{n}(t)=\frac{1-4(\al+\la)^{2}-8r(s)}{16}\left(1+\mathcal{O}\left(n^{-\frac{1}{2}}\right)\right),
\eq
where $s=4nt$ and the above formula is uniform for $t\in (0,c],$ $c$ is a positive constant. With the definition of $H_{n}(t)$ in (\ref{B12}) and integrating once again,
then one derives the formula (\ref{B7}) and the Theorem 1 is proved.

\section{Proof of Theorem 3}

The asymptotic behaviors of the leading coefficients $\gamma_{n}(t)$ can be derived from the asymptotic behavior of $Y(z)$ for large $z.$
\par
{\bf Proof of Theorem 3.} Taking a series inverse transformations from the concrete formulas (\ref{B37}), (\ref{B40}), (\ref{B43}) and (\ref{B53}), then one obtains the expression of $Y(z)$ as follows
\bq\label{B74}
Y(4nz)=(4n)^{\left(n+\frac{\al+\la}{2}\right)\sigma_{3}}e^{\frac{n\ell}{2}\sigma_{3}}R(z)P^{(\infty)}(z)e^{\left(g(z)-\frac{\ell}{2}\right)\sigma_{3}}\left(4n\right)^{-\frac{\al+\la}{2}\sigma_{3}},\;\; z\in \mathbb{C}\setminus [0,+\infty).
\eq
From the explict expression of $P^{(\infty)}(z)$ in (\ref{B43}), then it follows its expansion for large $z$ as follows
\bq\label{B75}
P^{(\infty)}(z)=I-2^{-(\al+\la)\sigma_{3}}\left(\frac{\al+\la}{4}\sigma_{3}+\frac{1}{4}\sigma_{2}\right)2^{(\al+\la)\sigma_{3}}\frac{1}{z}+\mathcal{O}\left(\frac{1}{z^{2}}\right).
\eq
By the definition of $g(z)$ in (\ref{B38}), it is easy to verify that
\bq\label{B76}
e^{ng(z)\sigma_{3}}=\left(I-\frac{n}{4z}\sigma_{3}+\mathcal{O}\left(\frac{1}{z^{2}}\right)\right)z^{n\sigma_{3}}.
\eq
Moreover, by the estimation of (\ref{B56}) and the formula (\ref{B57}), one obtains,
\bq\label{B77}
R(z)=I+\mathcal{O}\left(n^{-\frac{1}{2}}z^{-1}\right),
\eq
for large $n,$ $z\rightarrow \infty$ and uniformly for $t\in (0,c],$ $c$ is a finite positive constant.
\par
Combining the above asymptotic expansions, then one finds,
\bq\label{B78}
Y(4nz)=I+\frac{Y_{1}}{\widehat{z}}+\mathcal{O}\left(\frac{1}{\widehat{z}^{2}}\right),
\eq
where $\widehat{z}=4nz,$ $Y_{1}$ is independent of $z$ and its expression is given as follows
\bq\label{B79}
Y_{1}=4n\left(
\begin{matrix}
-\frac{n+\al+\la}{4}+\mathcal{O}\left(n^{-\frac{1}{2}}\right)& i4^{2n-1}n^{2n+\al+\la}e^{n\ell}\left(1+\mathcal{O}\left(n^{-\frac{1}{2}}\right)\right)\\
-i4^{-2n-1}n^{-(2n+\al+\la)}e^{-n\ell}\left(1+\mathcal{O}\left(n^{-\frac{1}{2}}\right)\right)& \frac{n+\al+\la}{4}+\mathcal{O}\left(n^{-\frac{1}{2}}\right)\\
\end{matrix}
\right).
\eq
With the aid of the general result of (3.30) in \cite{DeiftBook}, one gets,
\bq\label{B80}
\left(\gamma_{n}(t)\right)^{2}=i\left(2\pi[Y_{1}]_{12}\right)^{-1}=\frac{1}{\pi}2^{-4n-1}n^{-(2n+\al+\la+1)}e^{n\ell}\left(1+\mathcal{O}\left(n^{-\frac{1}{2}}\right)\right),
\eq
uniformly for $t\in(0,c]$ and $c$ is a finite positive constant.
\par
Recalling the quantity $[Y(0)]_{11}[Y(0)]_{12}$ in (\ref{B70}), the above approximation of $\left(\gamma_{n}(t)\right)^{2}$ and (\ref{B29}), then one obtains,
\bbq\label{B81}
R_{n}(t)=-2r'(s)\left(1+\mathcal{O}\left(n^{-\frac{1}{2}}\right)\right),
\eeq
this is (\ref{B15}). With above equation and (\ref{B34}), it follows (\ref{B17}) immediately.
\par
Obviously, (\ref{B73}) gives (\ref{B16}). From equations (\ref{B72}), (\ref{B73}) and (\ref{B35}), it is easy to verify (\ref{B18}).
\par
Substituting (\ref{B70}) and (\ref{B80}) into the right hand side of (\ref{B32}), then one follows,
\bq\label{B82}
\frac{d}{dt}\ln{\gamma_{n}(t)}=-r'(s)\left(1+\mathcal{O}\left(n^{-\frac{1}{2}}\right)\right).
\eq
Integrating the above equation on both sides, one gets the equation as follows
\bq\label{B83}
\ln{\frac{\gamma_{n}(t)}{\gamma_{n}(0)}}=\frac{r(s)-r(0)}{4n}\left(1+\mathcal{O}\left(n^{-\frac{1}{2}}\right)\right),
\eq
with the boundary conditions of $r(0)$ in (\ref{B2}), $\gamma_{n}(0)=\left(\frac{D_{n}(0; \al, \la)}{D_{n+1}(0; \al, \la)}\right)^{\frac{1}{2}},$ and $D_{n}(0; \al, \la)$ in (\ref{B8}), then it follows (\ref{B19}) immediately.
\par
At last, by the transform of (\ref{B0001}), $s=4nt$ and substituting (\ref{B15}) into the finite $n$ ${\rm P_{\uppercase\expandafter{\romannumeral5}}}$ equation in (\ref{B13}), then it follows the third-order equation (\ref{B20}). Similarly, inserting (\ref{B72}) into the $\sigma-$form ${\rm P_{\uppercase\expandafter{\romannumeral5}}}$ equation in (\ref{B14}), one obtains the equation (\ref{B21}).
\par
It completes the proof of Theorem 3.

\par

{\bf Acknowledgement.}
\par
This work of En-Gui Fan was  supported by the National Science Foundation of China under Project No.11671095.
\\

     \end{document}